%% file: main.tex
\documentclass{iccps}

\usepackage{subfigure}
\usepackage{amssymb, amsmath, amsfonts}
\usepackage{tikz}
\usepackage{enumitem}
\usepackage{cite}
\usepackage{pgfplots}
\usepackage{tikz}
\pgfplotsset{every axis/.append style={font=\scriptsize}}
\usetikzlibrary{calc,plotmarks,shapes}

\newtheorem{proposition}{Proposition}
\newtheorem{theorem}{Theorem}
\newtheorem{definition}{Definition}
\newtheorem{lemma}{Lemma}

\newtheorem{remark}{Remark}

\newlength\figureheight
\newlength\figurewidth
\DeclareMathOperator*{\argmin}{arg\; min}     

\newcommand{\I}{{{\rm I}}}

\newcommand{\abs}[1]{\left|#1\right|}
\newcommand{\ie}{{\it i.e.}}
\newcommand{\norm}[1]{\left\lVert#1\right\rVert}
\newcommand{\T}{^{\top}}
\newcommand{\D}{\partial}
\newcommand{\sgn}[1]{{\rm sgn}(#1)}

\title{Robust State Estimation against Sparse Integrity Attacks}

\numberofauthors{3}
\author{
  \alignauthor
  Duo Han\\
  \affaddr{School of EEE}\\
  \affaddr{Nanyang Technological University}\\
  \affaddr{Singapore, 639798}\\
  \email{dhanaa@ntu.edu.sg}
  \alignauthor
  Yilin Mo\\
  \affaddr{School of EEE}\\
  \affaddr{Nanyang Technological University}\\
  \affaddr{Singapore, 639798}\\
  \email{ylmo@ntu.edu.sg}
  \alignauthor
  Lihua Xie\\
  \affaddr{School of EEE}\\
  \affaddr{Nanyang Technological University}\\
  \affaddr{Singapore, 639798}\\
  \email{elhxie@ntu.edu.sg}
}
\date{\today}

\begin{document} \maketitle

\begin{abstract}
We consider the problem of robust state estimation
in the presence of integrity attacks. There are $m$ sensors
monitoring a dynamical process. Subject to the integrity attacks, $p$ out of $m$ measurements can be arbitrarily manipulated. The classical
approach such as the MMSE estimation in the literature may not provide
a reliable estimate under this so-called $(p,m)$-sparse attack. In this work, we propose a robust estimation framework where distributed local measurements are computed first and fused at the estimator based on a convex optimization problem. We show the sufficient and necessary conditions for robustness of the proposed estimator. The sufficient and
necessary conditions are shown to be tight, with a trivial gap.
We also present an upper bound on the damage an attacker can cause when the sufficient condition is satisfied. Simulation results are also given to illustrate the effectiveness of the estimator.
\end{abstract}

\input{intro.tex}

\section{Problem Setup}\label{section:problemsetup}
\begin{figure}
  \centering
  \begin{tikzpicture}[block/.style={rectangle,rounded corners,draw=black,fill=blue!20,thick}]

    \node[block,label=below:$x(k)$] (process) at (3.0,-6) {Process};

    \node[block] (sensora) at (0.2,-4) {Sensor 1};
     edge [<-,semithick] (processa);
    \draw [semithick,<-] (sensora)--node[midway,left]{$y_1(k)$}(0.2,-6)--(process);

    \node[block] (sensorb) at (3.0,-4) {Sensor 2};
    \draw [semithick,<-] (sensorb)--node[midway,left]{$y_2(k)$}(3,-5)--(process);

    \node[text width=2cm] at (5.3,-4) {$\cdots$};

    \node[block] (sensorn) at (6.1,-4) {Sensor m};
    \draw [semithick,<-] (sensorn)--node[midway,left]{$y_m(k)$}(6.1,-6)--(process);

    \node[cloud, thick, cloud puffs=15.7, cloud ignores aspect, minimum width=3cm, minimum height=1cm, align=center, draw] (network) at (3,-2.5) {Network};

    \node[block] (attacka) at (1.3,-5) {Attack};
    \node[block] (attackb) at (4.1,-5) {Attack};

    \draw [semithick,->] (attacka)--(1.3,-4)--(sensora);
    \draw [semithick,->] (attackb)--(4.1,-4)--(sensorb);

    \draw [semithick,->] (sensora)--node[midway,left]{$\tilde x_1(k)$}(0.2,-3)--(0.2,-2.5)--(network);
    \draw [semithick,->] (sensorb)--node[midway,left]{$\tilde x_2(k)$}(3,-3)--(network);
    \draw [semithick,->] (sensorn)--node[midway,left]{$\tilde x_m(k)$}(6.1,-3)--(6.1,-2.5)--(network);


    \node[block] (scheduler) at (3,-1) {Remote Estimator} edge [<-,semithick] (network);

    \end{tikzpicture}
    \caption{System Block Diagram.}
\label{fig:model}
    \end{figure}
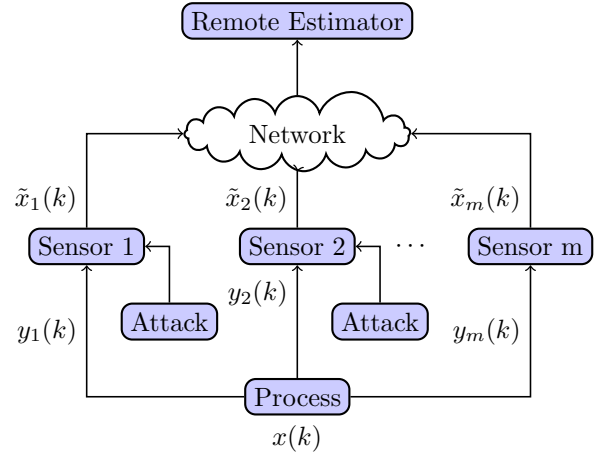

\subsection{System Model}
Assume that $m$ homogenous sensors are measuring the following LTI system (see Fig. \ref{fig:model}):
\begin{align}\label{eq:system}
 x(k+1) = Ax(k) + w(k).
\end{align}
The measurement equation for the $i$th sensor is given by
\begin{align}\label{eq:goodsystem}
 y_i(k) = C x(k) + \varepsilon_i(k),
\end{align}
where $x(k)\in\mathbb R^n$ is the state, $y_i(k)\in \mathbb R^{l}$ is the measurement collected by the $i$th sensor, $w(k) \in \mathbb R^{n}$ and $v(k) \in \mathbb R^{l}$ are the process noise and measurement noise for the $i$th sensor, respectively. The noise $w(k)$ and $\varepsilon_i(k)$'s are Gaussian distributed, \ie, $$w(k)\sim\mathcal N(0,Q),~\varepsilon_i(k)\sim\mathcal N(0,R).$$ The noises are assumed to be independent from each other across different time instants and sensors. Denote the tall measurement matrix $H\triangleq [C\T,C\T,\ldots,C\T]\T\in\mathbb R^{ lm \times n}$ and $y(k)\triangleq[y_1(k)\T,y_2(k)\T,\ldots,y_m(k)\T]\T$.. Denote $\Sigma = {\rm diag}(R,\ldots,R)$. The initial state $x(0)$ is Gaussian distributed with mean $\mu_0$ and variance $P_0$, and is independent from all noises. Assume that $(A,C)$ is observable and $(A,Q^{\frac{1}{2}})$ is controllable.

Kalman filter is well known as the recursive minimum mean square error (MMSE) estimator:
\begin{align*}
\hat x_{KF}(k) &= (A-K(k)HA)\hat x_{KF}(k-1) + K(k)y(k),\\
P^-(k) &= A P(k-1) A\T  + Q,\\
P(k) &= (\I_n-K(k)H)P^-(k),
\end{align*}
where the Kalman gain is given by
\begin{align}
K(k) = P^-(k) H\T (H P^-(k) H\T + \Sigma)^{-1}.\label{eq:k}
\end{align}
The state error covariance $P^-(k)$ converges exponentially fast to $\overline P$ which is obtained by solving the following discrete algebraic Riccati equation (DARE):
\begin{align}
X= A X A\T - A X H\T (H X H\T + \Sigma)^{-1} H X A\T + Q. \label{eq:dare}
\end{align}
Therefore, we assume the Kalman filter to be in the steady state, \ie, $P(k)=(\I_n-KH)\overline P$ and $K(k)=K$ from \eqref{eq:k}.

Due to the homogeneousness of the sensors, we know that $K=[G,\ldots,G]\T~,G\in\mathbb R^{n\times l}$. The Kalman filter can be equivalently rewritten as:
\begin{align}
\hat x_{KF}(k) = \frac{1}{m}\sum_{i\in\mathcal S}\tilde x_i(k),\label{eq:lse}
\end{align}
where
\begin{align}
\tilde x_i(k) &= (A-KHA)\tilde x_i(k-1) + m Gy_i(k),\label{eq:localestimate}
\end{align}
This means the estimation process at the estimator can be decomposed into $m$ sub-processes each of which only involves measurements from one sensor. This decomposition renders distributed estimation possible. To be specific, the sensor can locally compute $\tilde x_i(k)$ based on its own measurements and then the information fusion of all local estimates occurs at the remote estimator. It is worth noting that such distributed estimation is more resilient to attacks than the centralized estimation (all sensors transmit raw measurements). Since each local estimate of one sensor encodes all its historical measurements, corruption of one local estimate at some time instant causes little damage to the estimation.

Even if the sensor lacks computational capability and can only transmit raw measurements, each local estimation process can be computed at the central estimator. Therefore, without loss of generality, we assume each sensor computes a local estimate based on \eqref{eq:localestimate} and sends it to the estimator (see Fig. \ref{fig:model}).



\subsection{Attack Model}
The attacker launches an integrity attack to the sensory data in different fashions. For example, it can change the physical environment to mislead the sensors or it hacks the onboard sensor chip or it can manipulate the data packet during the sensor-to-estimator transmission. No matter in which way the attack is launched, we have the following equation:
\begin{align}\label{eq:badsystem}
  z_i(k) = \tilde x_i(k) + a_i(k),
\end{align}
where $z_i(k) \in \mathbb R^{n}$ is the ``manipulated'' local estimate and $a_i(k)\in \mathbb R^{n}$ is the attack vector. In other words, the attacker can change the local estimate of the $i$th sensor by $a_i(k)$. Define the local estimation error as $e_i(k)\triangleq {x}_k - \tilde x_i(k)$. Then we have
\begin{align}
  z_i(k) = x(k) +e_i(k) + a_i(k).
\end{align}
For concise notations, denote
\begin{align}
  \tilde x(k)&\triangleq[\tilde x_1(k)\T,\tilde x_2(k)\T,\ldots,\tilde x_m(k)\T]\T,\\
  z(k)&\triangleq[z_1(k)\T,z_2(k)\T,\ldots,z_m(k)\T]\T, \\
  e(k)&\triangleq[e_1(k)\T,e_2(k)\T,\ldots,e_m(k)\T]\T, \\
  a(k)&\triangleq[a_1(k)\T,a_2(k)\T,\ldots,a_m(k)\T]\T.& \nonumber
\end{align}

Denote the index set of all sensors as $\mathcal S\triangleq\{1,2,\ldots,m\}$. For any index set $\mathcal I\subseteq \mathcal S$, define the complement set to be $\mathcal I^c\triangleq \mathcal S\backslash\mathcal I$. In our attack model, we assume that the attacker can only compromise at most $p$ sensors but can arbitrarily choose $a_i(k)$. The index set of malicious sensors is assumed to be time invariant. Formally, a $(p,m)$-sparse attack can be defined as
\begin{definition}[$(p,m)$-sparse attack]
  A vector $a$ is called a $(p,m)$-sparse attack if there exists an index set $\mathcal I\subset \mathcal S$, such that:
  \begin{enumerate}[label=(\roman*)]
  \item $\norm{a_i(k)} = 0,~\forall i\in \mathcal I^c ;$
  \item $\abs{\mathcal I} \leq p.$
  \end{enumerate}
\end{definition}
Define the collection of a possible index set of malicious sensors as
\[
  \mathbb C\triangleq\{\mathcal I:\mathcal I\subset\mathcal S,\abs{\mathcal I} = p\}.
\]
The set of all possible $(p,m)$-sparse attacks is denoted as
\[
  \mathcal A =\mathcal A(k) \triangleq \bigcup_{\mathcal I\in\mathbb C}\{a(k): \norm{a_i(k)} = 0, i\in \mathcal I^c\},\forall k.
\]

After introducing the $(p,m)$-sparse attack, we need to formally define the robustness.
\begin{definition}[Robustness]
  An estimator $$g:\mathbb R^{mn}\mapsto \mathbb R^n$$ which maps the measurements $z(k)$ to a state estimate $\hat x(k)$ is said to be robust to the $(p,m)$-sparse attack if it satisfies the following condition:
  \begin{align}
    \norm{g(\tilde x(k))-g(\tilde x(k)+a(k))}\leq \mu(\tilde x(k)),~\forall a\in\mathcal A,\label{eq:defrobust}
  \end{align}
  where $\mu:\mathbb R^{mn}\mapsto \mathbb R$ is a real-valued mapping on $\tilde x(k)$.
\end{definition}
The robustness implies that the disturbance on the state estimate caused by an arbitrary attack is bounded. A trivial robust estimator is $g(y)=0$ which provides a very poor estimate. Therefore, another desirable property for an estimator is translation invariance, which is defined as follows:
\begin{definition}[Translation invariance]
An estimator $g$ is translation invariant if $g(z+Eu)=u+g(z),~\forall u\in\mathbb R^n$, where $E \triangleq [\I_n,\ldots,\I_n]\T$.
\end{definition}

\begin{remark}
  Notice that if an estimator is robust and translation invariant, then
  \begin{align*}
    &\| g(\tilde x(k)) - g(\tilde x(k) + a(k) ) \|\\
    &= \|E x(k)+ g(e(k)) - E x(k) + g(e(k)+a(k)) \| \\
    &= \|g(e(k)-g(e(k)+a(k))\|\leq \mu(e(k)).
  \end{align*}
  Therefore, the maximum bias that can be injected by an adversary is only a function of the noise $e(k)$.
\end{remark}

\subsection{A Robust Estimator}
Apparently, the linear estimator~\eqref{eq:lse} cannot give an estimate with bounded error even when only one estimate is arbitrarily manipulated. In other words, there is a conflict between the MMSE optimality and the robustness against attacks.

The main task of this work is to design a robust estimator which achieves a desirable tradeoff between the MMSE optimality and the robustness, and investigate the sufficient and necessary conditions to be robust to $(p,m)$-sparse attacks. To this end, a general estimator is proposed as follows:
  \begin{align}
    \hat x(k) \triangleq g(z(k))= \argmin_{\hat x(k)} \sum_{i\in\mathcal S} \varphi(z_i(k) -  \hat x(k)), \label{eq:generalestimator}
  \end{align}
where $\varphi:\mathbb R^n\mapsto \mathbb R$. We notice that to recover Kalman filter we can choose $\varphi$ to be $L_2$ norm. The candidate functions of $\varphi$ may include $L_p$ norm or LASSO \cite{tibshirani1996regression}, to just name a few.

Though there are many important estimators as special cases of \eqref{eq:generalestimator}, we mainly focus on the properties of the following concrete estimator in the rest of this paper. The same methodology can be extended to other $\varphi$'s. Pajic et al.~\cite{Pajic2014} proposed the following robust estimator in the presence of integrity attack:
  \begin{align*}
    & \mathop{\textit{minimize}}\limits_{\hat x(k),a,e(k)}&
    & \sum_{i\in\mathcal S}\norm{e_i(k)}_2^2\\
    &\text{subject to}&
    &z_i(k) = \hat x(k) + e_i(k)+a_i(k), \forall i,\\
    &&& a\in \mathcal A.
  \end{align*}
  However, the minimization problem involves zero-norm, and thus is difficult to solve in general. A commonly adopted approach is to use $L_1$ relaxation to approximate zero-norm, which leads to the following minimization problem:
  \begin{align}
    & \mathop{\textit{minimize}}\limits_{\hat x(k),a,\varpi(k)}&
    & \sum_{i\in\mathcal S}\norm{\varpi_i(k)}_2^2+\lambda \sum_{i\in\mathcal S}\norm{a_i(k)}_1\label{eq:optlasso}\\
    &\text{subject to}&
    & z_i(k) = \hat x(k) + \varpi_i(k)+a_i(k), \forall i.\nonumber
  \end{align}
  If we define the following function $F:\mathbb R^n \mapsto \mathbb R$:
  \begin{align}
    F(u)~\triangleq ~&\mathop{\textit{minimize}}\limits_{v\in\mathbb R^n}&
    &  \norm{u-v}_2^2 + \lambda  \norm{v}_1, \label{eq:lasso2}
  \end{align}
  then one can easily prove that the optimization problem \eqref{eq:optlasso} can be rewritten as
  \begin{align}
    \hat x(k) \triangleq g(z(k))= \argmin_{\hat x(k)} \sum_{i\in\mathcal S} F(z_i(k) -  \hat x(k)).\label{eq:lasso}
  \end{align}

In the next section, we shall present sufficient and necessary conditions for the robustness of the estimator \eqref{eq:lasso}. For concise notation, we will omit the time index $k$ if it is clear from the context.

\section{Robust Analysis}\label{section:main}
We provide an answer to the following question in this section: in what condition the proposed estimator in \eqref{eq:lasso} satisfies the robustness requirement \eqref{eq:defrobust}?

Before preceding to the main results, we give an explicit form of $F(u)$ given in \eqref{eq:lasso2}.
We can decompose $F(u)$ by letting $F(u) = \sum_{i=1}^n f(u_i)$, where $u_i$ is the $i$th entry of $u$ and $f(\tau):\mathbb R \mapsto \mathbb R$ is given by
    \begin{align}
    f(\tau)~\triangleq ~&\mathop{\textit{minimize}}\limits_{v\in\mathbb R}&
    &  (\tau-v)^2 + \lambda \abs{v}, \label{eq:lassoscalar}
  \end{align}
We define the RHS of \eqref{eq:lassoscalar} as
  \begin{align*}
    \pi(v)~\triangleq~ (\tau-v)^2 + \lambda \abs{v} .
  \end{align*}
Applying the KKT conditions, we know that
  \begin{align*}
 0\in \partial \pi(a^*) = -2\tau+2v^*+\sgn{v}\lambda.
  \end{align*}
Since $\pi(v)$ is not differentiable at $v=0$, by calculating the subgradient we have that $$v^*=0,~\textrm{if } \abs{\tau}\leq \frac{\lambda}{2}.$$
For $v^* \neq 0$, by letting
$$0 = -2\tau+2v^*+\sgn{v^*}\lambda,$$
we obtain that
\begin{align*}
v^* =  \left\{
             \begin{array}{ll}
               \tau-\frac{\lambda}{2}, & \textrm{if } \tau > \frac{\lambda}{2}, \\
               \tau+\frac{\lambda}{2}, & \textrm{if } \tau < - \frac{\lambda}{2}.
             \end{array}
           \right.
\end{align*}
Therefore, we have $f$ explicitly written as:
  \begin{align}
    f(\tau) = \left\{
             \begin{array}{ll}
               \tau^2, & \textrm{if } \abs{\tau} \leq \frac{\lambda}{2}, \\
               \lambda \abs{\tau}-\frac{\lambda^2}{4}, & \textrm{if } \abs{\tau}> \frac{\lambda}{2}.
             \end{array}
           \right.\label{eq:f}
  \end{align}

In the next proposition we present some useful properties of $f$ and $F$.
\begin{proposition}\label{prop:property}
The properties of $f$ and $F$ are summarized as follows:
\begin{enumerate}[label=(\roman*)]
\item $f$ and $F$ are convex.
\item $f$ and $F$ are symmetric, \ie, $f(u) = f(-u)$.
\item $f$ and $F$ are non-negative and $f(0) = 0$.
\item $f$ and $F$ are twice differentiable.
\end{enumerate}
\end{proposition}
The results are easy to verify and omitted here.

To obtain the sufficient and necessary conditions for robustness, we first need to show some findings on the derivative of $F$. To facilitate the analysis, we define two functions. For all $u,v\in \mathbb R^{n}$ and $t\in \mathbb R$, define $h:\mathbb R^n\times\mathbb R^n\times\mathbb R\mapsto \mathbb R$ as follows:
\[
h(u,v,t) \triangleq F(v+tu).
\]
Define the mapping $\phi:\mathbb R^n\mapsto \mathbb R^n$ ,
\begin{align}
\phi(u) \triangleq \nabla F(u) = [\nabla f(u[1]),\ldots,\nabla f(u[n])]\T, \label{eq:gradient}
\end{align}
where
\begin{align}\label{eq:gr1}
\nabla f(u[i])=\left\{
             \begin{array}{ll}
               2\abs{u[i]}, & \textrm{if } \abs{u[i]} \leq \frac{\lambda}{2}, \\
               \sgn{u[i]}\lambda, & \textrm{if } \abs{u[i]}> \frac{\lambda}{2},
             \end{array}
           \right.
\end{align}
where $\sgn{\cdot}$ is defined as: if $s = \sgn{v}$, then
\[
s[i]=\left\{
             \begin{array}{ll}
               +1, & \textrm{if } v[i]\geq 0, \\
               -1, & \textrm{if } v[i] < 0.
             \end{array}
           \right.
\]
Notice that a useful equality in the sequel is
\[
\frac{\D h(u,v,t)}{\D t} = \phi(v+tu)\T u.
\]

\begin{lemma}\label{lemma:1}
  The following statements are true:
  \begin{enumerate}[label=(\roman*)]
  \item The limit below is well defined for all $u\in\{u\in \mathbb R^{n}:\norm{u}<\infty\}$, \ie,
  \begin{align}
    C(u) \triangleq \lim_{t\rightarrow\infty} \frac{\D h(u,0,t)}{\D t} = \lambda \norm{u}_1,
    \label{eq:sublinear}
  \end{align}

  \item The following pointwise limit holds:
    \begin{align}
      \lim_{t\rightarrow\infty} \frac{\D h(u,v,t)}{\D t}= C(u).
      \label{eq:deltalimit}
    \end{align}
    Moreover, the convergence is uniform on any compact set of $(u,v)$.

    \item For any $u,v$, we have that
    \begin{align}
      \phi(v+u)\T u \leq C(u).
      \label{eq:deltaw}
    \end{align}

  \end{enumerate}
\end{lemma}
\begin{proof}
\begin{enumerate}[label=(\roman*)]
\item It is easy to see that
    \begin{multline*}
    \lim_{t\rightarrow\infty} \frac{\D h(u,0,t)}{\D t} \\
    = \lim_{t\rightarrow\infty} \phi(tu)\T u=\lambda \sum_{i=1}^n \sgn{u[i]}u[i] =  \lambda \norm{u}_1.
    \end{multline*}
\item We have that
    \begin{align*}
    \lim_{t\rightarrow\infty} \frac{\D h(u,v,t)}{\D t}=\lim_{t\rightarrow\infty}\phi(v+tu)\T u=\lambda \norm{u}_1.
    \end{align*}
    Due to the convexity of $F$, $\D h(u,v,t)/ \D t$ is monotonically non-decreasing with respect to $t$. Furthermore, $C(u)$ is continuous since it is a norm. Therefore, by Dini's theorem~\cite{rudin1964principles}, $\D h(u,v,t)/ \D t$ converges uniformly to $C(u)$ on a compact set of $(u,v)$.

\item From \eqref{eq:gradient}, we know that
    \[
    \phi(v+u)\T u \leq  \sum_{i=1}^n \lambda u[i]\leq \lambda \norm{u}_1.
    \]
    Therefore, we conclude that $\phi(v+u)\T u \leq C(u)$ for any $u,v$.

\end{enumerate}

\end{proof}

\begin{remark}
Intuitively speaking, one can interpret $F$ as a
potential field and the derivative of $F$ as the force generated by each sensor. By \eqref{eq:deltaw}, we know that the force from the potential field $F$ along the $u$ direction cannot
exceed $C(u)$. On the other hand, Equation \eqref{eq:deltalimit} implies that this bound is achievable.
\end{remark}

We are now ready to give the sufficient condition for the robustness of the estimator.
\begin{theorem}[Sufficient condition]
  If the following conditions hold:
    \begin{align}
      2p < m,\label{eq:sufficiency}
    \end{align}
  then the estimator $g$ is robust.
  \label{theorem:sufficient}
\end{theorem}

\begin{proof}
  Our goal is to prove that there exists a $\beta(\tilde x)$, such that for any $t> \beta(\tilde x)$, $\|u\| = 1$, $a\in\mathcal A$, the following inequality holds:
  \begin{align}
    \sum_{i\in\mathcal S}\frac{\D h(-u,z_i,t)}{\D t} > 0.
    \label{eq:sufficientdiff}
  \end{align}
  As a result, any point $\|\hat x\| > \beta(\tilde x)$ cannot be the solution of the optimization problem since there exists $\epsilon>0$ such that $(\|\hat x\|-\epsilon)\hat x/\|\hat x\|$ is a better point. Therefore, we must have $\|g(z)\|\leq \beta(\tilde x)$ and hence the estimator is robust.

  To prove \eqref{eq:sufficientdiff}, we will first look at benign sensors. We can always find a finite constant $N_i$ depending on $\delta$ and $\tilde x_i$ such that for all $t \geq N_i(\delta,\tilde x_i)$, the following inequality holds:
  \begin{align}
    \frac{\D h(-u,z_i,t)}{\D t} \geq C(u) - \delta = \lambda -\delta,
    \label{eq:deltaapprox2}
  \end{align}
  for any $\|u\|=1$.
  We define $\beta(z) \triangleq \max_{1\leq i\leq m} N_i(\delta,z_i)$ and fix $\delta$ to be
  \begin{align}
    \delta = \frac{(m-2p)\lambda}{m}.\label{eq:defdelta}
  \end{align}
    Hence, for $i = 1,\dots,m$, if $t > \beta_{\delta}(z)$ we know that
  \begin{align}
    \sum_{i\in \mathcal I^c} \frac{\D h(-u,z_i,t)}{\D t}  \geq  (m-p) \left(\lambda-\delta\right),\,\forall \|u\|=1. \label{eq:suff2}
  \end{align}
  We now consider malicious sensors. By Lemma \ref{lemma:1} (iii), we know that for $i\in\mathcal I$, and any $u$
  \begin{align*}
    \phi(z_i - tu)\T tu &= \phi(z_i - 2tu + tu)\T tu\leq C(tu)\\
    \Rightarrow \phi(z_i - tu)\T u & \geq -\lambda.
  \end{align*}
  Then we have
  \begin{align}
    \sum_{i\in \mathcal I} \frac{\D h(-u,z_i,t)}{\D t}  \geq  -p \lambda,\,\forall \|u\|=1. \label{eq:suff1}
  \end{align}
  Hence from \eqref{eq:defdelta},~\eqref{eq:suff1} and~\eqref{eq:suff2}, we know that
  \begin{align*}
    \sum_{i\in\mathcal S} \frac{\D h(-u,z_i,t)}{\D t} \geq  (m-p) \left(\lambda-\delta\right)-p \lambda> 0,
  \end{align*}
  which proves \eqref{eq:sufficientdiff}.
\end{proof}
It is shown that if the number of malicious sensors is less than the good sensors, then the estimator is robust. The intuition is that the sum force injected by any $p$ sensors from the potential field $F$ along the $u$ direction must be able to be balanced by the sum force of the rest $m-p$ sensors, \ie, zero-sum. Otherwise, the optimal estimate must lie in the infinity due to unbalanced driving forces along $u$ and thus violates the robustness defined in \eqref{eq:defrobust}.

We next present a necessary condition for the robustness of the estimator.
\begin{theorem}[Necessary Condition]
  If the following condition is satisfied:
  \[2p > m,\]
  then the estimator is not robust to the attack.
  \label{theorem:necessity1}
\end{theorem}
\begin{proof}
  The robustness of the estimator is equivalent to that the optimal estimate $\hat x$ satisfies $\norm{\hat x}\leq \mu(z)$ for all $a\in \mathcal A$, where $\mu$ is a real-valued function. To this end, we will prove that for any $r > 0$, there exists a $y$ such that all $\hat x$ that satisfies $\norm{\hat x}\leq r$ cannot be the optimal solution of \eqref{eq:lasso}.

  We will first look at the compromised sensors. For every $\delta>0$ we can always find a finite constant $N_i(\delta)$ such that for any $\hat x\in\{\hat x:\norm{\hat x}\leq r\}$ and for all $t > N_i$, the following inequality holds:
  \begin{align}
    \frac{\D h(u,z_i-\hat x,t)}{\D t} \geq C(u) - \delta \label{eq:nece1}
  \end{align}
  The inequality is due to the uniform convergence of $h(u,v,t)$ to $C(u)$ on $\{u\}\times \{v:v=z_i - \hat x,\,\|x\|\leq  r\}$.

  Let us choose
  \begin{align*}
    \delta = \frac{2p-m}{m}C(u),
  \end{align*}
  and $t = \max_{i\in\mathcal I} N_i(\delta)$ and $z_i = tu$ for all $i\in \mathcal I$, then we know for any $\|\hat x\|\leq  r$,
  \begin{align*}
    \sum_{i\in\mathcal I} \frac{\D h(u,z_i-\hat x,t)}{\D t} \geq pC(u) - p\delta.
  \end{align*}
  Now let us look at the benign sensors. By Lemma \ref{lemma:1} (iii) we have
  \begin{align}
    \frac{\D h(u,\tilde x_i-\hat x,t)}{\D t} \geq -C(u).\label{eq:nece2}
  \end{align}
  From \eqref{eq:nece1} and \eqref{eq:nece2},
  \begin{multline*}
    \sum_{i\in\mathcal S} \frac{\D h(u,z_i-\hat x,t)}{\D t} \geq (m-p)C(u)-pC(u)+p\delta > 0
  \end{multline*}
  Thus for such a $z_i$ satisfying
  \[ y_i=
    \left\{
      \begin{array}{ll}
        \tilde x_i, & \hbox{if } i\in\mathcal I^c\\
        tu, & \hbox{if } i\in\mathcal I,
      \end{array}
    \right.
  \]
  $\hat x+u$ is a better estimate than all $\hat x$ satisfying $\norm{\hat x}\leq  r$. Since $r$ is an arbitrary positive real number, we can conclude that the estimator is not robust.
\end{proof}

\section{Performance Analysis}\label{section:performanceanalysis}
In the previous section we have studied the robustness of the estimator. Now we focus our attention on the performance of the proposed estimator. We concern two questions in this section. The first one is the sufficient condition that the estimator gives an MMSE estimate when there is no attack. The other one is what is the maximum damage that an attacker can cause to the estimate, \ie, the upper bound of $\norm{g(\tilde x(k))-g(\tilde x(k)+a(k))}$.
\subsection{Without attacks}
When no attacks are present, an MMSE estimate like a Kalman filter provides is still prefered. Notice that the proposed robust estimator indeed probabilistically provides an MMSE estimate. A sufficient condition for providing the MMSE estimate $\hat x_{KF}$ given in \eqref{eq:lse} is given as follows.
\begin{lemma}\label{lemma:mmsecondition}
If $\tilde x \in \mathcal G$, where
\begin{align}
 \mathcal G \triangleq \{\tilde x \in \mathbb R^{mn}:\max_{i\in\mathcal S} \norm{\tilde x_i - \hat x_{KF}}_1\leq \frac{\lambda}{2}\},\label{eq:mmsecondition}
\end{align}
then $\hat x = \hat x_{KF}$.
\end{lemma}
\begin{proof}
From \eqref{eq:mmsecondition} and \eqref{eq:f}, we know that $\hat x_{LS}$ is a solution of \eqref{eq:lasso}.
\end{proof}

Now we characterize the pdf of $\tilde x$. Define the local estimation error covariance of the $i$th sensor and the local cross estimation error covariance between the $i$th sensor and the $j$th sensor as
\begin{align*}
  P_{ii}(k)&\triangleq \mathbb{E}[e_i(k)e_i(k)\T|y_i(1),\ldots,y_i(k)],\\
  P_{ij}(k)&\triangleq \mathbb{E}[e_i(k)e_j(k)\T|y_i(1),y_j(1),\ldots,y_j(k),y_j(k)].
\end{align*}
From \eqref{eq:localestimate}, the error dynamics of the $i$th sensor estimate is thus given as follows:
\begin{multline}
e_i(k) = (A-KHA)e_i(k-1)\\
+(m GC-\I_n)w(k)+mG\varepsilon_i(k).
\end{multline}
Note that the local estimator for each sensor is a stable estimator since the spectral radius of $A-KHA$ is less than one \cite{anderson79}. It is easy to see that $P_{ii}(k)$ converges to $\overline{P}_{ii}$ at the steady state, where $\overline{P}_{ii}$ is the unique solution of the following Lyapunov equation of $X$:
\begin{multline}
X=(A-KHA)X(A-KHA)\T\\
+(mGC-\I_n)Q(mGC-\I_n)\T+m^2GRG\T.
\end{multline}
Similarly, $P_{ij}(k)$ converges to $\overline{P}_{ij}$, where $\overline{P}_{ij}$ is the unique solution of the following Lyapunov equation of $X$:
\begin{multline}
X=(A-KHA)X(A-KHA)\T\\
+(mGC-\I_n)Q(mGC-\I_n)\T.
\end{multline}
Denote $\Gamma=\{\overline P_{ij}\}\in \mathbb R^{nm\times nm}$.
Now we know the probability density function of $\tilde x$, \ie,
\[
\tilde x \sim \mathcal N(x , \Gamma),
\]
and thus the distribution of $\hat x_{KF}$. We can compute the probability of generating the MMSE estimate
\begin{align}
\Pr\left(\tilde x\in \mathcal G\right) = \int_{x\in \mathcal G}\mathcal N(x, \Gamma){\rm d}\tilde x.
\end{align}
The integration is not trivial and numerical methods can be used to approximate $\Pr\left(\tilde x\in \mathcal G\right)$. A closed-form solution to $\Pr\left(\tilde x\in \mathcal G\right)$ is left as an open question.

Another interesting observation is that the larger $\lambda$ is, the more likely the MMSE estimate is.

\subsection{Under attacks}
We now consider the worst damage that an attacker can cause, \ie, the maximum deviation between the estimate under attacks and that without attacks. If the necessary condition in Theorem \ref{theorem:necessity1} is violated, the estimator is not robust and thus the deviation can be arbitrarily large. A more interesting question is how to obtain $\mu(\tilde x)$ in \eqref{eq:defrobust} for all possible attacks if the estimator is robust.

Suppose the sufficient condition in Theorem \ref{theorem:sufficient} is satisfied. Let the robust estimate without attacks to be $\hat x_{R} = g(\tilde x)$. Due to the translation invariance, we have
\begin{multline*}
\norm{g(\tilde x)-g(\tilde x+a)}_1\\
 = \norm{\hat x_{R}-g(\tilde x+a)}_1 = \norm{g(z - E\hat x_{R})}_1 \leq \mu(\tilde x).
\end{multline*}
Denote $\tilde z_i \triangleq z_i - \hat x_{R}$, $\tilde z= [\tilde z_1,\ldots,\tilde z_{m}]$, and $\breve x_i \triangleq \tilde x_i - \hat x_{R}$, $\breve x= [\breve x_1,\ldots,\breve x_{m}]$.

Similar to the proof of Theorem \ref{theorem:sufficient}, there exists $\beta^*$ such that for any $\beta\in\{\beta:\norm{\beta}_1>\norm{\beta^*}_1\}$ the following inequality holds:
  \begin{align}
    \sum_{i\in\mathcal S}\phi(\tilde z_i - \beta )\T \sgn{\beta} > 0 \label{eq:forcerequirement}
  \end{align}
In other words, we want to find a $\beta^*$ such that
  \begin{align}
    \sum_{j=1}^n\sum_{i\in\mathcal S}\nabla f\left(\tilde z_i[j] - \beta^*[j] \right)\sgn{\beta^*[j]}= 0,~\forall j=1,\ldots,n, \label{eq:force2}
  \end{align}
where $\tilde z_i[j]$ and $\beta^*[j]$ are the $j$th entry of $\tilde z_i$ and $\beta^*$ respectively.

Define the two mapping $\underline\kappa,\overline\kappa: \mathbb R^m\times\mathbb R\times\mathbb R \mapsto \mathbb R$ for any vector $u$ and scalars $p,m$:
\begin{align*}
&\underline\kappa(u,p,m) \triangleq \\
&~~~~\left\{u[i]:\abs{\{u[j]:u[j]\leq u[i],j\neq i\}} = \left\lfloor\frac{m-2p}{2}\right\rfloor+1\right\},\\
&\overline\kappa(u,p,m) \triangleq \\
&~~~~\left\{u[i]:\abs{\{u[j]:u[j]\geq u[i],j\neq i\}} = \left\lfloor\frac{m-2p}{2}\right\rfloor+1\right\}.
\end{align*}
Let $\zeta_j\triangleq [\breve x_1[j],\ldots,\breve x_m[j]]\T$. Then we denote $(\underline \theta_i,\overline \theta_i)$ to be
\begin{align}
(\underline \theta_j,\overline \theta_j)=(\underline\kappa(\zeta_j,p,m),\overline\kappa(\zeta_j,p,m)),~j=1,\ldots,n. \label{eq:groupmedian}
\end{align}
Now we are ready to present the upper bound on the worst damage.
\begin{theorem}\label{theorem:upper}
The upper bound $\mu(\tilde x)$ is shown as follows :
\begin{align}
\mu(\tilde x) = \norm{\beta^+}_1,
\end{align}
where $\beta^+_j = \max \left(\abs{\underline \theta_j-\lambda/2}, \abs{\overline \theta_j+\lambda/2}\right),~i=1,\ldots,n$.
\end{theorem}
\begin{proof}
A sufficient condition for $\eqref{eq:force2}$ is that for each $j$ the following inequality holds:
  \begin{align}
    \sum_{i\in\mathcal S}\nabla f\left(\tilde z_i[j] - \beta[j] \right)\sgn{\beta^*[j]} = 0.  \label{eq:force3}
  \end{align}
We first show that $\beta^*[j]$ must lie in $[\underline \theta_j-\lambda/2, \overline \theta_j+\lambda/2]$. We prove this by contradiction. Suppose $\beta^*[j]< \underline \theta_j-\lambda/2$. For any possible $\mathcal I^c$, we then have
\[
    \sum_{i\in\mathcal I^c}\nabla f\left(\breve x_i[j] - \beta[j] \right)\sgn{\beta^*[j]} \geq (m-p)\lambda.
\]
This is due to the facts that there are at least $m-p$ points of $\tilde z_i[j]$ are $\lambda/2$ larger than $\beta^*[j]$ from \eqref{eq:groupmedian} and that the maximum gradient is $\lambda$ from \eqref{eq:gr1}.

On the other hand, from Lemma \ref{lemma:1} (iii) we know that for any possible $\mathcal I$,
\[
    \abs{\sum_{i\in\mathcal I}\nabla f\left(\tilde z_i[j] - \beta[j] \right)\sgn{\beta^*[j]}} \leq p\lambda.
\]
Hence \eqref{eq:force3} cannot hold for $\beta^*[j]< \underline \theta_j-\lambda/2$. Similar arguments applies for $\beta^*[j]> \overline \theta_j+\lambda/2$. Therefore, we know that $\beta^*[j]\in[\underline \theta_j-\lambda/2, \overline \theta_j+\lambda/2]$. If we take the maximum over $(\abs{\underline \theta_j-\lambda/2}, \abs{\overline \theta_j+\lambda/2})$ as $\beta^+[j]$, then any $\beta$ satisfying $\norm{\beta}>\norm{\beta^+}$ cannot be $g(z - E\hat x_{R})$.
\end{proof}

\section{Simulation Results}\label{section:simulation}
In this section we illustrate the main results using numerical examples.

Consider a linear system with
\begin{align*}
A&=\begin{bmatrix} 0.95& 1\\ 0& 1.01\end{bmatrix},Q=\begin{bmatrix} 1.5 &1\\1 &2\end{bmatrix}
\end{align*}
is monitored by $m=5$ sensors with
\begin{align*}
C&=\begin{bmatrix} 1 &0\\0 &1\end{bmatrix},R=\begin{bmatrix} 2 &1\\1 &1\end{bmatrix}.
\end{align*}
First we verify the sufficient and necessary conditions for robustness. Assume the number of the malicious sensor is $p=2$. In Fig. \ref{fig:1} we depict the trajectory of the true state and show that the proposed robust estimators give reliable estimates with bounded error.
\begin{figure}
  \centering
  \input{trajwattacks.tikz}
  \caption{Trajectory of the system state and robust estimate with different $\lambda$. The number of malicious sensors is $2$ out of $5$. When $p>m/2$, the estimate goes unbounded. The robust estimator with $\lambda=10$ performs worst.}\label{fig:1}
\end{figure}
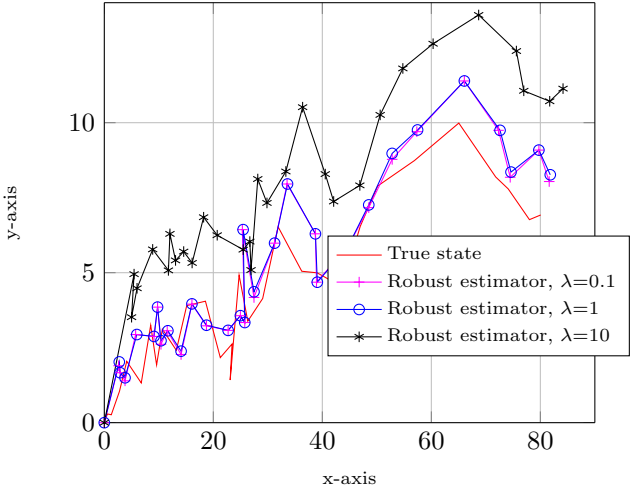
As comparison, we also plot the case when there are no attacks in presence in Fig. \ref{fig:2}. Still all the estimators compute reliable estimates. Notice that the estimator with $\lambda=10$ resembles the Kalman filter most. This collides with the intuition that a large penalty parameter $\lambda$ performs poorly in Fig. \ref{fig:1} but works well without attacks.

In Table. \ref{table:comparison1} we show the relationship between the penalty parameter $\lambda$ and the probability of recovering Kalman filter when there is no attack.
On the other hand, we plot the upper bound $\mu(\tilde x)$ given in Theorem \ref{theorem:upper} and the true gap $\norm{g(\tilde x) -g(z)}_1$ versus time. Notice that when $\lambda = 1$ or $\lambda = 0.1$, the upper bound on the deviation caused by attacks is smaller. In other words, the estimator is more robust with a small $\lambda$. Tradeoff between robustness when the sensors are under attack and the MMSE optimality when the attacker is not present is clearly shown via different $\lambda$'s.
\begin{table}[]\label{table:comparison1}
\centering
\begin{tabular}{|l|c|c|c|c|c|}
  \hline
  $\lambda$ & 1 & 2 & 5 & 10 \\
  \hline
  $\Pr(\hat x = \hat x_{KF})$ & 0.0001  &  0.013  &  0.48 &   0.98  \\
  \hline
\end{tabular}
\caption{Relationship between the penalty parameter $\lambda$ and the probability of recovering Kalman filter when there is no attack.}
\end{table}

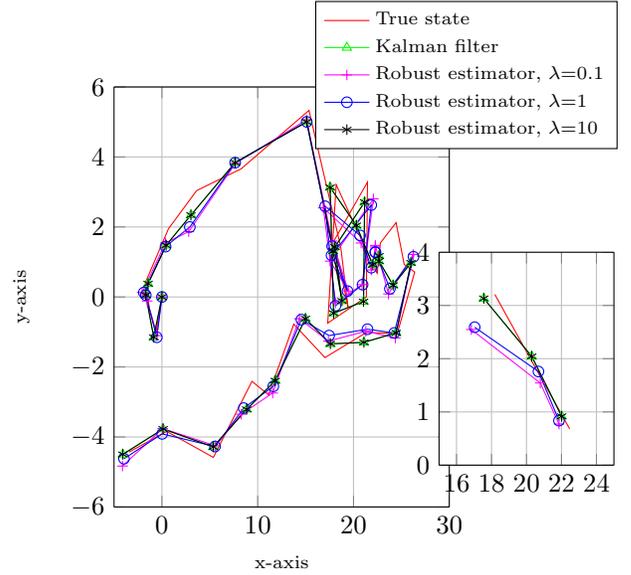
\begin{figure}
  \centering
  \input{trajwoattacks.tikz}\\
  \caption{Trajectory of the system state and robust estimate with different $\lambda$ without attacks. The robust estimator with $\lambda=10$ resembles the Kalman filter most.}\label{fig:2}
\end{figure}

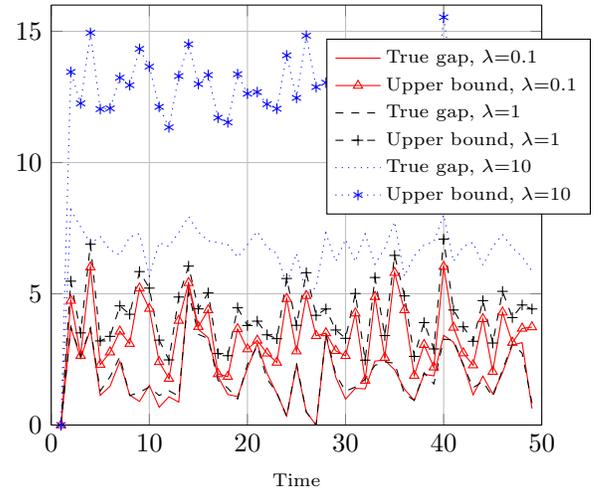
\begin{figure}\label{fig:3}
    \centering
    \input{upperbound.tikz}
    \caption{Upper bound $\mu(\tilde x)$ and the true gap versus time. The number of malicious sensors is $2$ out of $5$. The robust estimator with $\lambda=10$ has the largest upper bound and the true gap.}
\end{figure}
\section{Concluding Remarks}\label{section:conlusion}

In this work we have studied the robust state estimation problem in the presence of integrity attacks. The attacker can control $p$ out of $m$ sensors and can arbitrarily change the measurement. We have proposed a robust estimation framework and formulated a convex optimization problem with $L_1$ regulation to find the robust estimate. We have also shown the sufficient and necessary conditions for the estimator is robust against the $(p,m)$-sparse attack. Informally speaking, the percentage of compromised sensors should be less than a half to guarantee the robustness. Furthermore, we have analyzed the estimation performance without attacks and under attacks. Further work includes the robust estimation with inhomogeneous sensors.

\bibliographystyle{IEEEtran}
\bibliography{reference}
\end{document}

%% file: intro.tex
\section{Introduction}

Sensor networks have been increasingly applied in various cyber-physical systems (CPSs) such as smart grid~\cite{MassoudAmin2005} or Supervisory Control And Data Acquisition (SCADA) systems~\cite{boyer2002scada}. The sensors, however, are vulnerable to integrity attacks since in most cases they are spatially distributed and cannot be fully protected. Typically, the adversary can control a portion of all sensors and arbitrarily change their measurements during attacks. The objectives for launching such an attack in industrial systems may include using free electricity in smart grid, stealing resources like gasoline from oil caverns, causing economical loss for rivals, etc. Two famous attacks on CPS hampering 
the critical infrastructure
are Maroochy Water incident \cite{slaylessons} and the first SCADA system malware (called Stuxnet) \cite{chen2011lessons}. To sum up, Security in control and estimation systems has received much research attention~\cite{challengessecurity}.

In this article, we focus on the problem of robust state estimation based on compromised sensory data. The classical approach such as Kalman filtering cannot generate a reliable estimate in the presence of attacks. To put it simply, taking all measurements as important, e.g., (weighted) averaging all the measurements is not a good idea for robust estimation since one large measurement will drive the final estimate far deviated from the true value. To be concrete, we consider the problem of estimating the state $x\in\mathbb R^n$ of a dynamical process from measurements collected by $m$ homogenous sensors, where the measurements are subject to Gaussian noise. Integrity attacks are very likely because the sensors cannot be fully protected due to practical reasons such as high maintenance cost. We assume the attacker can only take control of up to $p<m$ sensors since the resources of attacker may be limited and some sensors are physically untouchable. We put no restriction on what the attack is like. In other words, once a sensor is attacked, the sensory data can be arbitrarily manipulated.

\emph{Related Work}: In the context of power systems, the estimation based on irregular sensor data has been formulated as bad data detection problem~\cite{Handschin1975,Mili1985}. A common practice is identifying the bad data or outliers by checking the corresponding residue. But this does not work well for intentional attacks \cite{liu2009,henrik2010,Xie2011,Kim2014}. For example, Liu et al.~\cite{liu2009} showed that it is possible to launch a stealthy attack without being noticed. On the top of that, a so-called framing attack was studied in~\cite{Kim2014}. The detector is very likely to abandon the critical measurement under the framing attacks. Without the critical measurements, the network is unobservable and a stealthy attack is possible.

The robust estimator has been long studied in the field of statistics~\cite{Hampel1974 ,Kassam1985,robust2006,robust2009,Yohai2012,Rousseeuw1992,Hossjer1994,Rousseeuw2012}. The robustness is often quantified by breakdown points, e.g., the percentage of bad sensors, beyond which the estimate is unstable~\cite{Hampel1971,donoho1983notion}. However, the application of robustness analysis has not been extended to the estimation problem of a dynamical system yet.

For dynamical systems, compromised data detection via fault detection and isolation based methods has been extensively studied, \cite{fp-ab-fb:09b,Pasqualetti2011,wirelesscontrol,Fawzi2012,chong2015observability}. However, in most of these works, the system is assumed to be noiseless, which greatly favors the failure detector. Pajic et al.~\cite{Pajic2014} extended \cite{Fawzi2012} by taking the bounded system noise into account. They proved that the worst error is still bounded for all possible attacks once the sufficient condition for exact data recovery in noiseless systems is satisfied. The drawback is that their approach is involved with zero-norm and thus computationally intractable, especially for large scale systems. In \cite{Mo2010,moscs10security}, the authors studied the worst bias an attack can cause through reachability analysis and ellipsoid approximation.

The significance of this work is provide a robust estimation framework for estimating noisy systems compared with the noiseless case in \cite{Fawzi2012}. To mitigate the damage injected by the attacker, we propose a robust estimator based on the convex optimization problem involving $L_1$ regulation which takes an advantage of analytical simplicity over \cite{Pajic2014}. The proposed estimator is shown to provide a robust estimate under some sufficient condition. Furthermore, the upper bound of the gap between the estimate without attacks and that under attacks is also quantified.

The rest of the paper is organized as follows. In Section \ref{section:problemsetup} we formulate the robust estimation problem. We study the sufficient and necessary conditions for robustness in Section \ref{section:main}. We analyze the estimation performance in Section \ref{section:performanceanalysis}. Simulations are illustrated in Section \ref{section:simulation}. The concluding remarks are given in Section \ref{section:conlusion}. 

\textit{Notations}: The $i$th entry of the vector $u$ is denoted as $u[i]$. The $L_p$ norm of the vector $u$ is denote as $\norm{u}_p$. If unspecified, $\norm{u}$ means the $L_2$ norm of $u$ by default. $\lfloor v\rfloor$ means the largest integer that is less than the scalar $v$.

%% file: trajwattacks.tikz
%
%
\definecolor{mycolor1}{rgb}{1.00000,0.00000,1.00000}%
\begin{tikzpicture}

\begin{axis}[%
width=0.951\figurewidth,
height=\figureheight,
at={(0\figurewidth,0\figureheight)},
scale only axis,
separate axis lines,
every outer x axis line/.append style={black},
every x tick label/.append style={font=\color{black}},
xmin=0,
xmax=90,
xlabel={x-axis},
xmajorgrids,
every outer y axis line/.append style={black},
every y tick label/.append style={font=\color{black}},
ymin=0,
ymax=14,
ylabel={y-axis},
ymajorgrids,
axis background/.style={fill=white},
legend style={at={(0.455,0.158)},anchor=south west,legend cell align=left,align=left,fill=white}
]
\addplot [color=red,solid]
  table[row sep=crcr]{%
0	0\\
0.563593982089259	0.291607900635238\\
1.28511687176477	0.265059168625969\\
2.78605219194118	1.0462243089063\\
4.11887838321942	2.04595708021059\\
6.78509299798206	1.32447961101334\\
8.51110357194653	3.25143500492472\\
9.59764578268494	1.92696585265584\\
10.6569690223932	3.1186828002842\\
13.2198820113325	2.4130470173581\\
15.1239718107327	3.90290655283148\\
18.5571212707547	4.04972894266936\\
21.2833735130931	2.16954449136775\\
23.4400435928125	2.62371672398486\\
23.0957444705295	1.42234818363005\\
24.733983065843	4.9145664058965\\
26.3283276996664	3.38480117878789\\
29.0591457813576	4.13914047417967\\
31.8919277341167	6.52813900932508\\
36.23245183304	5.04782796719241\\
39.0361581624189	4.99659923589803\\
43.6043961919294	4.53341511725076\\
46.9563313007125	6.54380095983253\\
50.1670088529806	7.89119970505644\\
56.8917557332188	8.73627623714645\\
65.043373317013	9.98723229421841\\
71.8115320631794	8.19016948059249\\
74.1090355208511	7.79230394529426\\
77.9998182804546	6.76807437533857\\
80.0651193721528	6.92444188426605\\
};
\addlegendentry{True state};

\addplot [color=mycolor1,solid,mark=+,mark options={solid}]
  table[row sep=crcr]{%
0	0\\
2.63010671193603	1.89842231745848\\
2.89951351419362	1.66651649462981\\
3.80597281455778	1.39574903053403\\
5.75248377481253	2.93897852894774\\
9.12589612527163	2.87481087012413\\
9.7757937204883	3.85210413390026\\
10.4014459060456	2.6136004951783\\
11.5471987908064	3.06406649237623\\
14.087222041979	2.27397296925214\\
16.0809483029248	3.95189470441117\\
18.7438061289515	3.2298762757815\\
22.7198618135831	3.07759668914527\\
24.9686507561559	3.56743178968642\\
25.7704915438217	3.33883844087518\\
25.3360606965194	6.43532272518511\\
27.4643827341727	4.17439024396703\\
31.2457905943113	5.98455875050882\\
33.5336761042473	7.9565891356588\\
38.7338341004319	6.29578549597073\\
38.9338613453747	4.68298185058857\\
44.7564825368819	5.88498811921063\\
48.5189863664573	7.16964958775319\\
52.8149467422571	8.77828807676258\\
57.3402477810908	9.70041479833007\\
66.056521621534	11.391389000202\\
72.3256179708269	9.7451626481436\\
74.4699688088352	8.18139919105478\\
79.7295773189273	9.08667765579388\\
81.6267950449689	8.03885135816535\\
};
\addlegendentry{$\text{Robust estimator, }\lambda\text{=0.1}$};

\addplot [color=blue,solid,mark=o,mark options={solid}]
  table[row sep=crcr]{%
0	0\\
2.74965937620006	2.02561223143135\\
2.89951430744346	1.66651481906643\\
3.80597421473048	1.49727235851776\\
5.96535242470196	2.9389692423604\\
9.12590006335146	2.87480946758377\\
9.77579657705985	3.85211024747923\\
10.4014495234265	2.73200345278215\\
11.6666717662974	3.06406466101188\\
14.0872194961584	2.38416270382386\\
16.0809529311218	3.96073172831112\\
18.7438063413926	3.24939639055768\\
22.7198645498445	3.07759389104072\\
24.9686510479549	3.5674334816613\\
25.7704896381019	3.33883524545275\\
25.4686605691014	6.43532233175969\\
27.4643834762527	4.36372798516222\\
31.2457833404626	5.98456130199577\\
33.6122404198003	7.95659443002704\\
38.7338311548063	6.2957836738702\\
39.0817831924508	4.68298441745686\\
44.7564776987733	5.88499046855282\\
48.5189860312876	7.25976321487395\\
52.8149454349911	8.97564358649989\\
57.4772666979751	9.75804247007979\\
66.0565202797469	11.3913895560815\\
72.6050074536119	9.74516288149708\\
74.6158420049793	8.35435465338997\\
79.7295743488808	9.08667639042793\\
81.8183121871118	8.26384949013047\\
};
\addlegendentry{$\text{Robust estimator, }\lambda\text{=1}$};

\addplot [color=black,solid,mark=asterisk,mark options={solid}]
  table[row sep=crcr]{%
0	0\\
5.4637370428356	4.95178469213459\\
5.00757938948159	3.51701508655062\\
6.00986109729143	4.48032431305114\\
8.87847913672756	5.77660637793202\\
11.7663823781848	5.07725237804017\\
12.0457922537195	6.29381554248664\\
13.0719799322574	5.40942232492604\\
14.5668891490567	5.7037605634844\\
16.1118685999039	5.32489563503918\\
18.2236915439209	6.84914187936265\\
20.718825921297	6.2494125036707\\
25.4841204947906	5.76807590787808\\
26.7043868014147	6.03710547726735\\
26.9218091903089	5.08886917901774\\
28.1586394144207	8.12231976894205\\
29.8455203078427	7.32493941690484\\
33.310549355058	8.37364150267799\\
36.3867246776269	10.5170960172147\\
40.5435465910528	8.28733096062756\\
42.0822456744131	7.37346227050845\\
46.8835700507737	7.91092266475607\\
50.6152481392479	10.259767540149\\
54.7747320331241	11.808723709997\\
60.3265445663136	12.6323457987858\\
68.6835271927643	13.5902199818493\\
75.6049969885824	12.3884228007155\\
76.961140002901	11.0632690037988\\
81.7186535404862	10.7140647977952\\
84.1625846061123	11.1359062671442\\
};
\addlegendentry{$\text{Robust estimator, }\lambda\text{=10}$};

\end{axis}
\end{tikzpicture}%

%% file: trajwoattacks.tikz
%
%
\definecolor{mycolor1}{rgb}{1.00000,0.00000,1.00000}%
\begin{tikzpicture}

\begin{axis}[%
width=0.65\figurewidth,
height=\figureheight,
at={(0\figurewidth,0\figureheight)},
scale only axis,
separate axis lines,
every outer x axis line/.append style={black},
every x tick label/.append style={font=\color{black}},
xmin=-5,
xmax=30,
xlabel={x-axis},
xmajorgrids,
every outer y axis line/.append style={black},
every y tick label/.append style={font=\color{black}},
ymin=-6,
ymax=6,
ylabel={y-axis},
ymajorgrids,
axis background/.style={fill=white},
legend style={at={(0.6,0.852)},anchor=south west,legend cell align=left,align=left,fill=white}
]
\addplot [color=red,solid]
  table[row sep=crcr]{%
0	0\\
-0.612764068181217	-0.996598957337363\\
-1.32597901661523	-0.300972613326265\\
-1.69182441124481	0.433100005019196\\
0.677605564530237	1.94659025029154\\
3.62727379888961	3.04066767335885\\
8.26644697703035	3.66077528108927\\
15.3511147061789	5.33027089286036\\
18.6659216422025	1.42809605816989\\
19.4260574592234	-0.289799248444669\\
17.9692809827141	1.60950483775244\\
21.4583926020674	3.28309665102656\\
21.3418475264901	0.000661488599175364\\
17.3068769699656	-0.747022848067549\\
18.1900960350742	3.20747543894641\\
20.4443047044209	1.86945129273047\\
22.4697517692217	0.680560212892868\\
22.7993198985576	1.54323041654102\\
24.4388190962842	2.12334768299246\\
25.2925105541618	0.926502483412979\\
26.3877772637628	0.71295467339072\\
24.5666508088409	-1.04559129556917\\
21.35412025203	-1.01304827866938\\
17.017897304737	-1.73176895394368\\
13.764895817478	-0.776599530811371\\
11.1397700522058	-2.79813274381321\\
9.39851330862106	-2.4103943835664\\
5.37049548403433	-4.58344681110831\\
0.501490619937627	-3.84147383650986\\
-3.43654560124666	-4.43292893937554\\
};
\addlegendentry{True state};

\addplot [color=green,solid,mark=triangle,mark options={solid}]
  table[row sep=crcr]{%
0	0\\
-0.855875334535381	-1.15151230799717\\
-1.68849581300371	0.054772715038625\\
-1.44334907159996	0.379544222949905\\
0.436612709066712	1.47595164173515\\
3.02813684734958	2.33945964503917\\
7.6349506266133	3.83716990469195\\
15.1173047883677	5.00341463399827\\
17.874244115688	1.3116039741988\\
18.7955092644714	-0.113035834248424\\
18.0497122147994	1.4255402506802\\
21.1547496338278	2.71112542712272\\
21.04813294783	-0.131317152073887\\
17.919321608223	-0.464339798905064\\
17.5580725591421	3.13166448089503\\
20.3004697544645	2.04358694208543\\
22.0111016496698	0.917968689750302\\
22.6643946705606	1.17184740007166\\
22.6532642282775	1.02625537249664\\
24.1356909358202	0.335880006463363\\
25.9943824780951	0.974509156487867\\
24.4557231796817	-1.02391693750498\\
21.0857817907297	-1.29663732338019\\
17.5831491928768	-1.33983567114067\\
14.9845499873997	-0.618016301489551\\
11.8268174647236	-2.39217219492215\\
8.85105630949238	-3.20965853939891\\
5.35638880201799	-4.28358640115931\\
0.123757757767491	-3.77102966374661\\
-4.08238426966086	-4.49676414529346\\
};
\addlegendentry{Kalman filter};

\addplot [color=mycolor1,solid,mark=+,mark options={solid}]
  table[row sep=crcr]{%
0	0\\
-0.547771976975029	-1.10666521368527\\
-1.53708114331171	-0.114407634344102\\
-2.06496329625851	0.127392215009162\\
0.377272220335077	1.54414247853445\\
2.79412547966537	1.8655740797985\\
7.66544867378734	3.81610304827051\\
15.0975099090737	5.06055585029212\\
17.767295783685	1.59870600366388\\
19.4639015627258	0.168247927055758\\
17.556284930983	1.02288663168816\\
22.0744261231339	2.80224481346277\\
20.8918489566653	0.35132217041161\\
18.1165548027645	-0.168780847791927\\
16.8469002136431	2.54792829432002\\
20.7891029955014	1.54396182363696\\
21.8585408822269	0.772062287152781\\
22.2771166514558	1.42865314405215\\
22.2794034753189	1.47132686437481\\
23.6465492951633	0.0852739562624886\\
26.2301868507578	1.21490451515315\\
24.3655831645151	-1.17219478872939\\
21.2849254478128	-0.978466771604357\\
17.3764712648289	-1.26056353011549\\
14.3475443002847	-0.621826642779564\\
11.5640587701536	-2.74375814155507\\
8.28340179704896	-3.34881402662713\\
5.5841779305105	-4.25120644706362\\
0.0757253344317095	-3.75836524493205\\
-4.11489638733815	-4.83023610625458\\
};
\addlegendentry{$\text{Robust estimator, }\lambda\text{=0.1}$};

\addplot [color=blue,solid,mark=o,mark options={solid}]
  table[row sep=crcr]{%
0	0\\
-0.53357152340695	-1.15151230566871\\
-1.6884958019945	0.052985279786204\\
-1.88921054751621	0.127397868996352\\
0.379202898435165	1.43907121287922\\
2.91348004670498	1.99890830011013\\
7.64196715457704	3.83716989776326\\
15.1096777363963	5.00341461418399\\
17.7672875561909	1.45312744152045\\
19.3729228709603	0.168248479225677\\
17.7030468418395	1.18987439139419\\
21.8494161785257	2.62619958843587\\
20.94552646031	0.351321375243753\\
18.1165530803294	-0.251934747122652\\
17.0372008332051	2.59130991296885\\
20.6832285455455	1.75872106619732\\
21.8585414435469	0.836867119093638\\
22.2771185125642	1.28328414667381\\
22.2836027491279	1.2751454998928\\
23.7961574441945	0.235048486942746\\
26.2616277794023	1.14814601775998\\
24.2518860595453	-1.02726562793334\\
21.4560274019102	-0.920756241651297\\
17.4484939070784	-1.10304505917981\\
14.5585054574269	-0.64682704069808\\
11.6243089029549	-2.55385816183849\\
8.54325531625391	-3.17431676274516\\
5.54398170534276	-4.27156580812434\\
0.0220396536830814	-3.91122535955435\\
-3.96039344133641	-4.62095444245848\\
};
\addlegendentry{$\text{Robust estimator, }\lambda\text{=1}$};

\addplot [color=black,solid,mark=asterisk,mark options={solid}]
  table[row sep=crcr]{%
0	0\\
-0.855897624917806	-1.15153264430067\\
-1.68849466126492	0.0547704102591463\\
-1.44333436952088	0.379539851184075\\
0.436613673858556	1.4759496185178\\
3.02814386208364	2.33946762671913\\
7.63495122758205	3.83717040448159\\
15.1173102096305	5.00342646022609\\
17.8742418359046	1.31159082947406\\
18.795512745436	-0.11303614352421\\
18.0497096923198	1.42552737541427\\
21.1547442431263	2.71112793418692\\
21.04812533529	-0.131315268098883\\
17.9193189237547	-0.464328927024038\\
17.5580777620034	3.13166668775766\\
20.3004682969819	2.04357457036854\\
22.0111009778815	0.917963942733629\\
22.66437535953	1.17182897644343\\
22.6532602841912	1.02624058711122\\
24.1356730323191	0.335859667322811\\
25.994372067984	0.974487011694522\\
24.4557262855649	-1.02392875484017\\
21.0857728916426	-1.29661575392636\\
17.5831488058195	-1.33982551913817\\
14.9845489602628	-0.61801076847351\\
11.8268174440765	-2.39217347017815\\
8.85105841554507	-3.20966302251646\\
5.3563894547949	-4.28358706914563\\
0.123759132797963	-3.77102658552871\\
-4.0823840393387	-4.49676387980745\\
};
\addlegendentry{$\text{Robust estimator, }\lambda\text{=10}$};

\end{axis}

\begin{axis}[%
width=0.339\figurewidth,
height=0.507\figureheight,
at={(0.63\figurewidth,0.1\figureheight)},
scale only axis,
separate axis lines,
every outer x axis line/.append style={black},
every x tick label/.append style={font=\color{black}},
xmin=15,
xmax=25,
xmajorgrids,
every outer y axis line/.append style={black},
every y tick label/.append style={font=\color{black}},
ymin=0,
ymax=4,
ymajorgrids,
axis background/.style={fill=white}
]
\addplot [color=red,solid,forget plot]
  table[row sep=crcr]{%
18.1900960350742	3.20747543894641\\
20.4443047044209	1.86945129273047\\
22.4697517692217	0.680560212892868\\
};
\addplot [color=green,solid,mark=triangle,mark options={solid},forget plot]
  table[row sep=crcr]{%
17.5580725591421	3.13166448089503\\
20.3004697544645	2.04358694208543\\
22.0111016496698	0.917968689750302\\
};
\addplot [color=mycolor1,solid,mark=+,mark options={solid},forget plot]
  table[row sep=crcr]{%
16.8469002136431	2.54792829432002\\
20.7891029955014	1.54396182363696\\
21.8585408822269	0.772062287152781\\
};
\addplot [color=blue,solid,mark=o,mark options={solid},forget plot]
  table[row sep=crcr]{%
17.0372008332051	2.59130991296885\\
20.6832285455455	1.75872106619732\\
21.8585414435469	0.836867119093638\\
};
\addplot [color=black,solid,mark=asterisk,mark options={solid},forget plot]
  table[row sep=crcr]{%
17.5580777620034	3.13166668775766\\
20.3004682969819	2.04357457036854\\
22.0111009778815	0.917963942733629\\
};
\end{axis}
\end{tikzpicture}%

%% file: upperbound.tikz
%
%
\begin{tikzpicture}

\begin{axis}[%
width=0.951\figurewidth,
height=\figureheight,
at={(0\figurewidth,0\figureheight)},
scale only axis,
separate axis lines,
every outer x axis line/.append style={black},
every x tick label/.append style={font=\color{black}},
xmin=0,
xmax=50,
xlabel={Time},
xmajorgrids,
every outer y axis line/.append style={black},
every y tick label/.append style={font=\color{black}},
ymin=0,
ymax=16,
ymajorgrids,
axis background/.style={fill=white},
legend style={at={(1.1,0.92)}, legend cell align=left,align=left,fill=white}
]
\addplot [color=red,solid]
  table[row sep=crcr]{%
1	0\\
2	3.71600999714018\\
3	2.54474748127076\\
4	3.65256529278335\\
5	1.133587433522\\
6	1.48004931254482\\
7	2.41935197356518\\
8	1.14265744568548\\
9	0.899683961580347\\
10	1.51302304875983\\
11	0.678493946315468\\
12	1.07657269212063\\
13	0.876399406619091\\
14	5.31735775688703\\
15	3.6455831923274\\
16	3.32613886402301\\
17	1.84688077051954\\
18	1.16480793008819\\
19	1.08660730664012\\
20	2.16994643357184\\
21	3.08921926499813\\
22	1.97136230882456\\
23	1.20027799319958\\
24	0.349149303190334\\
25	2.27502543847584\\
26	0.496074367021449\\
27	2.81214494624038e-06\\
28	3.40549393875993\\
29	1.82458034805707\\
30	0.984840033620887\\
31	1.40012428848166\\
32	1.37946329885731\\
33	2.44018486722579\\
34	2.46804024174206\\
35	2.1026275463152\\
36	1.37086052204203\\
37	0.939868896928864\\
38	1.93596173179118\\
39	1.87361515001482\\
40	3.41646992565852\\
41	3.15618357763734\\
42	2.31706593388846\\
43	1.15234179930777\\
44	1.86177700989554\\
45	1.18275658164981\\
46	1.99848562580133\\
47	3.04136813561396\\
48	3.14218614416845\\
49	0.629683979957345\\
};
\addlegendentry{$\text{True gap, }\lambda\text{=0.1}$};

\addplot [color=red,solid,mark=triangle,mark options={solid}]
  table[row sep=crcr]{%
1	0\\
2	4.74226513651853\\
3	2.64474605062024\\
4	6.01414630038382\\
5	2.30338453136357\\
6	2.77462994377639\\
7	3.56937472560009\\
8	3.10045502783648\\
9	5.21024419902431\\
10	4.43587977927524\\
11	2.40657481519958\\
12	1.77062898983271\\
13	3.97999623882655\\
14	5.41735841243526\\
15	3.74558176821883\\
16	4.38040342667324\\
17	1.94687512647312\\
18	1.84468320714429\\
19	3.66082996322268\\
20	2.89285662244449\\
21	3.22330273640185\\
22	2.73200497424809\\
23	2.38396409516331\\
24	4.8127116332072\\
25	2.81532542591582\\
26	4.92508239439172\\
27	3.39973789439051\\
28	3.50549640831362\\
29	2.82845240708128\\
30	2.63765106440258\\
31	4.259140072951\\
32	1.69771262801885\\
33	4.88689661434495\\
34	2.56804114802809\\
35	5.80421013544213\\
36	4.38045841109971\\
37	1.88472755534251\\
38	3.05578506236251\\
39	2.20302261845591\\
40	6.03449433709172\\
41	3.71652029596366\\
42	2.74112485340292\\
43	2.28243421905964\\
44	4.04193317130532\\
45	2.03488988223412\\
46	4.29570596115197\\
47	3.14136865243516\\
48	3.67923747069246\\
49	3.73247482825227\\
};
\addlegendentry{$\text{Upper bound, }\lambda\text{=0.1}$};

\addplot [color=black,dashed]
  table[row sep=crcr]{%
1	0\\
2	3.74215874926255\\
3	2.6206411515178\\
4	3.70335065745465\\
5	1.29779263103504\\
6	1.83975255243579\\
7	2.60646173059599\\
8	1.12516097199855\\
9	1.24237479324393\\
10	1.45134694658927\\
11	1.12592532469532\\
12	1.30722239327019\\
13	1.104915470212\\
14	5.05512932868869\\
15	3.45505823504567\\
16	3.29521721022945\\
17	1.7035982631735\\
18	1.40288524994448\\
19	0.992771042839791\\
20	2.31222033029229\\
21	3.01307387240222\\
22	1.75710614701559\\
23	1.23467068128992\\
24	0.317318713460769\\
25	2.359128128\\
26	0.498042600641269\\
27	3.25944777834764e-07\\
28	3.47226023984695\\
29	1.93168407894249\\
30	1.29120780569864\\
31	1.44901037893237\\
32	1.74543604409627\\
33	2.27361640387379\\
34	2.40955761821314\\
35	2.35330936666271\\
36	1.19367950918448\\
37	0.942817534300582\\
38	1.98969404804444\\
39	1.56608671062069\\
40	3.26913109356314\\
41	3.24524702654588\\
42	2.28351178472263\\
43	1.42148110839088\\
44	1.65988728772043\\
45	1.06906789332917\\
46	2.10100375045801\\
47	3.09490037002148\\
48	2.73795956563167\\
49	0.838865220630129\\
};
\addlegendentry{$\text{True gap, }\lambda\text{=1}$};

\addplot [color=black,dashed,mark=+,mark options={solid}]
  table[row sep=crcr]{%
1	0\\
2	5.48158329089758\\
3	3.51067976462807\\
4	6.88875275587333\\
5	3.20610073852451\\
6	3.37865978883501\\
7	4.54000155017844\\
8	4.22032408444676\\
9	5.84653913272099\\
10	5.21674458789834\\
11	3.2352339974807\\
12	2.4823686643292\\
13	4.87719357079943\\
14	6.05512599037128\\
15	4.42653243932926\\
16	5.03422389873968\\
17	2.72375589312431\\
18	2.6437148990057\\
19	4.46700286293957\\
20	3.79327862035482\\
21	3.95801302253246\\
22	3.43907314020829\\
23	3.29073906439244\\
24	5.58345906056283\\
25	3.79943844609656\\
26	5.80386237899749\\
27	4.17925406177677\\
28	4.42288740074627\\
29	3.62134419381178\\
30	3.30410682012405\\
31	5.01179397212752\\
32	2.47306080440076\\
33	5.62033608706926\\
34	3.40955573948129\\
35	6.46241878931912\\
36	4.92032518615231\\
37	2.614460823181\\
38	3.90205243171002\\
39	2.90017113660613\\
40	7.08183067316949\\
41	4.38345738856856\\
42	3.74292314573963\\
43	3.19066538062176\\
44	4.74004034355742\\
45	3.12541862262758\\
46	5.09318567300734\\
47	4.09489969331062\\
48	4.57378057474233\\
49	4.42329951611188\\
};
\addlegendentry{$\text{Upper bound, }\lambda\text{=1}$};

\addplot [color=blue,dotted]
  table[row sep=crcr]{%
1	0\\
2	8.20642581738377\\
3	7.546536356306\\
4	6.91933164018903\\
5	7.17406564489373\\
6	6.66629413925148\\
7	6.47438847806599\\
8	7.16307164020841\\
9	7.29124021992167\\
10	5.67790530157755\\
11	6.94775680164686\\
12	6.80293995044819\\
13	7.33996078208882\\
14	7.94314341186829\\
15	7.35082393766632\\
16	7.0240091460379\\
17	6.94779361851173\\
18	6.87096892019253\\
19	6.39458778702772\\
20	6.86764842446217\\
21	7.37554335107818\\
22	6.76608654024647\\
23	6.53735204249855\\
24	5.3388905205937\\
25	6.54217428264188\\
26	4.97127187830128\\
27	5.24999633494019\\
28	7.31466068710065\\
29	6.23530553431469\\
30	7.07415232268491\\
31	6.22574996934302\\
32	7.27859132361313\\
33	6.128240469458\\
34	6.73987941920158\\
35	7.74399943230839\\
36	5.66294018990385\\
37	6.45542778981079\\
38	6.86116151582519\\
39	6.97838501137822\\
40	8.00385907050142\\
41	6.23680666646095\\
42	6.80278334877185\\
43	6.95570692237414\\
44	6.0884221267992\\
45	6.80055794385917\\
46	7.26290492132911\\
47	6.76067610542217\\
48	6.41429959772226\\
49	5.86117420115808\\
};
\addlegendentry{$\text{True gap, }\lambda\text{=10}$};

\addplot [color=blue,dotted,mark=asterisk,mark options={solid}]
  table[row sep=crcr]{%
1	0\\
2	13.4524268743738\\
3	12.2537341017087\\
4	14.9418830979196\\
5	12.0348071203951\\
6	12.0625160203231\\
7	13.2342008408047\\
8	12.9480027284272\\
9	14.3321043536296\\
10	13.6544093957467\\
11	12.1221334077697\\
12	11.3422999295247\\
13	13.2944662066153\\
14	14.5073472736345\\
15	12.9914727612064\\
16	13.3352155576479\\
17	11.7082634223675\\
18	11.5300839102885\\
19	13.3679414121964\\
20	12.6239436805577\\
21	12.6853613773097\\
22	12.2241669723692\\
23	12.0519975093459\\
24	14.0842338929824\\
25	12.4634682849877\\
26	14.8336253696921\\
27	12.8752089878789\\
28	13.0446795671505\\
29	12.4871460840064\\
30	12.0437969200721\\
31	13.518095828193\\
32	11.4567083472446\\
33	13.6823274073803\\
34	12.2703248430781\\
35	14.2697478899785\\
36	13.4578897558047\\
37	11.4915474399973\\
38	12.5399473712061\\
39	11.7625236723669\\
40	15.5343486923361\\
41	12.8761680986521\\
42	12.4679468969882\\
43	12.0419737506302\\
44	13.4898023036479\\
45	11.9982313718776\\
46	13.8273335139988\\
47	12.6846652637262\\
48	13.0527623447834\\
49	13.0139361568078\\
};
\addlegendentry{$\text{Upper bound, }\lambda\text{=10}$};

\end{axis}
\end{tikzpicture}%